\documentclass[journal, a4paper, twoside]{IEEEtran}
\usepackage[margin=0.73 in]{geometry}
\usepackage{amsmath,amssymb,amsfonts,amsthm}
\usepackage{algorithmic}
\usepackage{graphicx}
\usepackage{textcomp}
\usepackage{bbm}
\usepackage{multirow}
\usepackage{diagbox}
\usepackage{adjustbox}
\usepackage{caption, multirow, makecell}
\setcellgapes{3pt}
\usepackage{mathtools}

\theoremstyle{definition}
\newtheorem{thm}{Theorem}
\newtheorem{lem}{Lemma}

\newtheorem{defn}{Definition}
\newtheorem{rem}{Remark}

\title{Optimal Demand Private Coded Caching for Users with Small Buffers}
\author{K.K. Krishnan Namboodiri and B. Sundar Rajan \\
Department of Electrical Communication Engineering, IISc Bangalore, India \\
E-mail: \{krishnank, bsrajan\}@iisc.ac.in \vspace{-3mm}
}

\begin{document}

\maketitle
\begin{abstract}
Coded Caching is an efficient technique to reduce peak hour network traffic. One limitation of known coded caching schemes is that the demands of all users are revealed to their peers in the delivery phase. Schemes that assure privacy for user demands are studied in recent past. Assuming that the users are equipped with caches of small memory sizes, the achievable rate under demand privacy constraints is investigated in this work. We present an MDS code based demand private coded caching scheme with $K$ users and $N$ files that achieves a memory rate pair $\left(\frac{1}{K(N-1)+1},N\left(1-\frac{1}{K(N-1)+1}\right)\right)$. The presented memory-rate pair meets the lower bound under demand-privacy requirements, proposed by Yan \textit{et al.} in the recent work \cite{c13}. By memory sharing this characterizes the exact rate-memory trade-off for the demand private coded caching scheme for cache memory $M\in \left[0,\frac{1}{K(N-1)+1}\right]$.
\end{abstract}


\section{Introduction}
\label{intrn}

Data traffic has grown at a rapid pace over the last decade especially that of multimedia files such as video-on-demand. Since communication resources become scarce during peak usage hours, caching at off-peak hours becomes a natural solution. Maddah Ali and Niesen had provided a scheme \cite{c1} that showed that by utilizing multicast opportunities, coded caching can achieve significant gain over uncoded caching. Yu, Maddah Ali and Avestimehr improved this scheme and provided one that is optimal for uncoded placement \cite{c2}. In \cite{c3}, Chen, Fan and Letaief showed that coded caching schemes with coded prefetching can perform better compared to schemes with uncoded prefetching.
The coded caching problem of rate minimization by the joint design of placement and delivery phase has been studied in different settings like decentralized caching\cite{c8}, shared caches\cite{c9}, multiple levels of caches\cite{c10} etc. By the design of placement delivery arrays(PDAs), Yan \textit{et al.} investigated the trade-off between rate and subpacketization of files\cite{c11}.    

In this work, we focus on the coded caching problem under demand privacy requirements of the users. Recently, demand privacy for the coded caching problem has been investigated in \cite{c4},\cite{c5},\cite{c6},\cite{c7} and \cite{c13}. In \cite{c7}, Wan and Caine study demand privacy for users with multiple file requests. In [$6$], Aravind \textit{et al.} address the problem of demand privacy under subpacketization requirements. In \cite{c4}, coded caching under perfect information theoretic privacy was studied. Sneha Kamath \textit{et al.} gives exact memory rate trade-off for a 2-user, 2-file demand private coded caching problem in \cite{c5}. In \cite{c13}, authors derive a lower bound on the rate for demand-private coded caching schemes. 

In this paper we present a demand private coded caching scheme that gives a new achievable $(M,R)$ pair for a general $K$-user, $N$-file system, which completely characterize the rate-memory trade-off of demand-private coded caching for cache memory $M\in [0,\frac{1}{K(N-1)+1}]$.   

\section{Preliminaries}

A server with $N$ files is connected to $K$ users through an error-free broadcast link. Assume that $N$ files $W_n; n\in [N]\triangleq \{1,2,\hdots,N\}$ are independent and each file is of size $F$ bits and takes values uniformly in the set $[2^F] = \{1,2,\hdots,2^F\}$. $\mathbf{W} \triangleq \{W_1,W_2,\hdots,W_N\}$. Each user $k\in[K]$ has a local cache of capacity $MF$ bits (Normalized capacity of $M$), $0\leq M\leq N$. Cache of the users are filled by the server with file contents in the placement phase when the user demands are unknown. The cache content of the $k^{th}$ user is denoted by $Z_k$. In the delivery phase, user $k$ demands file $W_{d_k}$ from the server, where demand $d_k$ are all i.i.d random variables distributed over $[N]$. The demands of all users is thus, $\textbf{d} = \{d_1,d_2,\hdots,d_K\}$. All users convey their demands over a private link between the user and the server, without revealing it to other users. Let $\mathbf{d}_{\tilde{k}}$ denote all demands but $d_k$. That is, $\mathbf{d}_{\tilde{k}} = \mathbf{d}\backslash \{d_k\}$.

Then the server broadcasts a message, $X^d$ of size $RF$ bits (normalized to $R$) to all the users.  $R$ is termed as the rate of the transmission.       

The main objective of coded caching problem is to minimize $R$ satisfying the decodabiity criterion 
\begin{equation}
H(W_{d_k}| Z_k, d_k, X^d) = 0, \hspace{0.5cm}  \forall k\in[K] \label{eq1}
\end{equation}
for all the users. That is, every user should be able to decode the demanded file using the local cache content and the transmission, $X^d$.

A coded caching scheme with $K$ users, $N$ files, cache size $M$ and rate $R$ is denoted as $(K,N;M,R)$ caching scheme or simply $(K,N)$ caching scheme. But in the usual non-private coded caching schemes, the demand of users will be revealed to their peers by the transmissions in the delivery phase. So the additional requirement of demand privacy for all users is,
\begin{equation}
I(\mathbf{d}_{\tilde{k}}\text{ } ;\text{ } d_k, Z_k, X^d | \mathbf{W} ) = 0.  \label{eq2}
\end{equation}  
The $k^{th}$ user should be completely uncertain about the demand of all other users, given all information user $k$ has. \\
A coded caching scheme that satisfy demand privacy constraint in  \eqref{eq2} in addition to \eqref{eq1} is referred as a demand private coded caching scheme. And the schemes that do not satisfy \eqref{eq2} are called as non-private coded caching schemes. For a $(K,N)$ non-private coded caching scheme for a cache memory of $M$, $R_{K,N}(M)$ denotes the minimum transmission rate $R$ by which the decodability criterion \eqref{eq1} is satisfied for all the users. $R^{*}_{K,N}(M)$ is the minimum achievable rate among all the $(K,N)$ non-private coded caching schemes for a cache memory of $M$. Similarly, for a cache memory $M$, $R_{K,N}^p(M)$ and $R^{p*}_{K,N}(M)$ are the achievable rate for a particular $(K,N)$ demand-private coded caching scheme and the minimum among all the $(K,N)$ demand private coded caching schemes respectively. That is, in addition to \eqref{eq1}, demand private schemes should satisfy \eqref{eq2} also. Therefore, any demand private coded caching scheme with transmission rate $R^{p}_{K,N}(M)$  for a cache memory $M$ will satisfy the following inequality.
\begin{equation}
R^{*}_{K,N}(M)\leq R^{p*}_{K,N}(M) \leq R^{p}_{K,N}(M) \label{eq3}
\end{equation}

Yan \textit{et al.} give a lower bound on the optimal rate-memory trade-off under demand privacy constraints in \cite{c13}. That is, a $(K,N)$ demand-private coded caching scheme with cache memory, $M$ should satisfy,
\begin{equation}
\label{lowerbound}
R^{p*}_{K,N}(M)\geq \max_{l\in[N]} \left(l+\frac{\min\{l+1,K\}(N-l)}{N-l+\min\{l+1,K\}}-lM\right).
\end{equation}

Consider a $(K,N)$ coded caching system. In a particular demand vector $\mathbf{d}$, let $l_n$ be the number of users requesting file $W_n,\text{ }n\in [N]$. We sort the set $\{l_1,l_2,\hdots,l_N\}$  in the descending order, and the resultant vector is termed as the demand profile of $\mathbf{d}$. In \cite{c12}, author studies the coded caching problem under different demand profiles. The term they use is \textit{demand type} instead of demand profile. For example, in a caching system with $4$ users and $2$ files, the possible demand profiles are $[4,0]$, $[3,1]$ and $[2,2]$.

\begin{defn}[\textbf{Uniform Demand Profile}]
If all the files are requested by the same number of users, then the resulting demand profile is said to be uniform. In the previous example, the demand profile $[2,2]$ is uniform.
\end{defn}


Now, we review some important results regarding demand-private coded caching that are already existing in the literature.
\begin{thm}(Theorem 1 in \cite{c4})
\label{thm1}
For a $(K,N)$ demand private coded caching with each user having a cache of memory, $M$, the rate
$$ R_{K,N}^{p}(M) = \frac{\binom{KN}{KM+1}-\binom{KN-N}{KM+1}}{\binom{KN}{KM}},\hspace{0.1cm} \text{if } M\in \{0,\frac{1}{K},\hdots,N\}$$
is achievable. Also, the lower convex envelope of $R_{K,N}^{p}(M)$ is achievable.
\end{thm}

\begin{thm}(Theorem 4 in \cite{c6})
\label{thm2}
If there exist a $(KN,N;M,R)$ non-private coded caching scheme, then there exist a $(K,N;M,R)$ demand private coded caching scheme.
\end{thm}

In \cite{c4}, author derives a $(K,N)$ demand private coded caching scheme from a $(KN,N)$ non-private scheme. By that construction, each of the $K$ users in the demand private scheme is mapped to $N$ virtual users in the non-private scheme (total of $KN$ virtual users).  For user $k\in [K]$ in the demand private scheme, server independently generates a number, $S_k$ uniformly at random from $[N]$ and conveys the same to user $k$ through the dedicated channel from the server to the user in the placement phase.  The cache placement is in accordance with the non-private scheme. But the earlier mapping of actual user in the demand private scheme to the virtual users in the non-private scheme gives $N$ different choices of the cache contents to every user connected to the server. So, during the placement phase, server populates  $k^{th}$ user's cache with the content that is supposed to keep in $(k-1)N+S_k$-th virtual user's cache. Therefore, under this placement policy user $k'$ will have an uncertainty regarding the cache content of user $k$ since $k'$ is unaware of $S_k$. In the delivery phase, the $K$ length demand vector $\mathbf{d}=[d_1,d_2,\hdots,d_K]$ is mapped to a $KN$ long virtual demand vector $\mathbf{\tilde{d}}=[\mathbf{q_1},\mathbf{q_2},\hdots,\mathbf{q_K}] $, where $\mathbf{q_k}$ is an $N$ length vector obtained by applying $(S_k-d_k) \text{ }mod \text{ }N$ right cyclic shifts on $(1,2,\hdots,N)$. The transmission is corresponding to the $(KN,N)$ non-private coded caching scheme for a demand vector $\mathbf{\tilde{d}}$. The user $k$ in the actual demand private scheme becomes the $(k-1)N+S_k$ -th user in the $(KN,N)$ non-private scheme and the user can decode the demanded file. Since $S_k$ is unknown for user $j\in[K]\backslash \{k\}$, privacy for all the users are guaranteed. More details can be found in \cite{c4}.

As the $(KN,N)$ non-private scheme, in \cite{c4} author uses the coded caching scheme proposed in \cite{c2} to achieve the rate memory trade-off in theorem \ref{thm1} where the placement is uncoded. In \cite{c5} and \cite{c6}, authors use coded placement for demand private coded caching. In our work, we propose a demand private coded caching scheme which also makes use of coded placement.

Each of the $\mathbf{q_1},\mathbf{q_2},\hdots,\mathbf{q_K}$ are just shifted versions of $(1,2,\hdots,N)$. So, by the construction of $\mathbf{\tilde{d}}$, each file is requested exactly by $K$ virtual users and hence the demand profile of the $(KN,N)$ non-private scheme used for the construction of $(K,N)$ demand-private scheme is always uniform. In Theorem \ref{thm2}, authors are giving a sufficient condition for the existence of a $(K,N)$ demand-private coded caching scheme. By the above observation regarding the virtual demand vector, we can improve the sufficient condition and restate Theorem \ref{thm2} as follows. 

\begin{thm} (Improved Theorem \ref{thm2})
\label{thm2a}
If an $(M,R)$ point is achievable for a $(KN,N)$ non-private coded caching scheme under a demand  with uniform profile, then the same is achievable for $(K,N)$ demand private coded caching scheme. That is, the rate achieved by the $(KN,N)$ non-private scheme for a demand with the demand profile $\underbrace{[K,K,\hdots,K]}_{\text{$N$ times}}$ can be achieved by the $(K,N)$ demand-private coded caching scheme with the same cache memory. 
\end{thm}

In \cite{c12}, it is showed that for a $(4,2)$ non-private coded caching scheme the rate $\frac{4}{3}$ is achievable under the demand profile $[2,2]$ at a memory $M=\frac{1}{3}$ which is unachievable for the demand profile $[3,1]$. But, we can construct a $(2,2)$ demand private scheme that achieves the rate $\frac{4}{3}$ at a cache memory of $M=\frac{1}{3}$ by the procedure given in \cite{c4}. This example motivates the general results presented in the subsequent sections.

\begin{rem}
\label{rem2}
The condition given in Theorem \ref{thm2} or the improved condition in Theorem \ref{thm2a} is not necessary for the existence of a $(K,N)$ demand-private coded caching scheme. In \cite{c13}, a different approach for the construction of demand-private coded caching scheme is adapted which achieves $(M,R)$ points that are not achievable by the virtual user scheme in \cite{c4} especially when $N>K$.  
\end{rem}

The contributions of this paper may be summarized as follows:
\begin{itemize}
\item The achievability of an $(M,R)$ pair $(\frac{1}{K(N-1)+1},N(1-\frac{1}{K(N-1)+1}))$ for a $(KN,N)$ non-private coded caching scheme under uniform demand profile is shown. 
\item For cache memory, $0\leq M\leq \frac{1}{K(N-1)+1}$ a $(K,N)$ demand-private coded caching scheme is introduced and proved its optimality.
\item Under demand privacy constraint \eqref{eq2}, the achievability of the rate $N(1-\frac{1}{K(N-1)+1})$ at a cache memory of $\frac{1}{K(N-1)+1}$ generalizes the achievability of the $(M,R)$ pair $(\frac{1}{3},\frac{4}{3})$ for a 2-user, 2-file demand-private coded caching scheme shown in \cite{c5}.
\end{itemize}

The main results of this paper are discussed in Section \ref{results}. In Section \ref{examples}, we show some examples of the proposed demand-private coded caching scheme.

\section{Main Results}
\label{results}

\begin{lem}
\label{lemma1}
For a $(KN,N)$ non-private coded caching system under a demand vector with uniform profile, the $(M,R)$ pair $\left(\frac{1}{K(N-1)+1},N(1-\frac{1}{K(N-1)+1})\right)$ is achievable for a sufficiently large field. Furthermore, for  $M\in [0,\frac{1}{K(N-1)+1}]$,
$$R_{KN,N}(M) = N(1-M)$$ is achievable.
\end{lem}

\begin{proof}
We present a $(KN,N)$ non-private coded caching scheme that achieves an $(M,R)$ pair $\left(\frac{1}{K(N-1)+1},N(1-\frac{1}{K(N-1)+1})\right)$ under a demand vector with uniform profile. That is, the server does the cache placement with a prior knowledge that the demand vectors will have a uniform profile. We are assuming that the files segments are from a finite field of large size. Also, $\oplus$ and $\ominus$ denote the field addition and subtraction respectively.
\subsection{Placement Phase}
File $W_n, \text{  } \forall n\in[N]$ is divided into $K(N-1)+1$ subfiles of equal size, $W_{n,j}, \text{ }j\in[K(N-1)+1]$. The normalized size of a subfile is, $|W_{n,j}| = \frac{1}{K(N-1)+1}$.
Encode the subfiles, $[W_{n,j}: j\in [K(N-1)+1]$ of the file $W_n$ with a $(KN,K(N-1)+1)$ MDS code, which results in $KN$  coded subfiles, each with a normalized size $\frac{1}{K(N-1)+1}$. The $i-th$ coded subfile of the file $W_n$ is denoted as $\mathcal{C}_{n,i}\text{ },\text{ }1\leq i \leq KN$.
\begin{itemize}
\item The cache of user $k\in[KN]$ is filled with $\bigoplus_{n=1}^N \mathcal{C}_{n,k}$. That is,
$$Z_k = \bigoplus_{n=1}^N \mathcal{C}_{n,k},\; k = 1,2,\dots,KN. $$
\end{itemize}
\subsection{Delivery Phase}
Let $\mathbf{d}_{KN\times1}$ be the demand vector and  $S_n$ be the set of all users demanding the file $W_n$. 
$$S_n\coloneqq \{k\in[KN]\text{  }: \text{   } d_k = n\} $$
We are assuming a uniform demand profile for $\mathbf{d}$. There are $N$ files and $KN$ users in the system. That is, each file is demanded exactly by $K$ users. Let $S_n = \{S_{n,1},S_{n,2},\hdots,S_{n,K}\}$. The content delivery is as follows,
\begin{itemize}
\item Corresponding to a user $k\in [KN]$, transmit $\mathcal{C}_{i,k}$, $\forall i\neq d_k.$ 
\end{itemize}
\subsection{decodability}
Claim: All the users can get their demanded file by the given transmission strategy.
\\
All users are requiring $K(N-1)+1$ subfiles of their respective demanded file. For a user $k \in S_n$, those subfiles  can be decoded by any $K(N-1)+1$ out of the $KN$ coded subfiles, $\mathcal{C}_{n,i}\text{ },\text{ }i \in [KN]$.
\begin{itemize}
\item  User $k$ directly receives the coded subfiles $\mathcal{C}_{n,j},\text{  }\forall j\in [KN]\backslash S_n$. That is, the user gets $KN-K$ coded subfiles of the file $W_n$.
\item User $k$ can calculate $\mathcal{C}_{n,k}$ by,
$$\mathcal{C}_{n,k} = Z_k \ominus \left( \underbrace{ \bigoplus_{i=1,i\neq n}^N \mathcal{C}_{i,k}}_{\text{Can calculate from transmissions }}\right)$$
\end{itemize}
That is, the user $k\in [KN]$ receives $K(N-1)+1$ coded subfiles of their respective demanded file. From those coded subfiles, the user can decode all the subfiles $W_{d_k,j}$, $\forall j\in [K(N-1)+1]$.

\subsection{Rate Calculation}
\begin{itemize}
\item Corresponding to each user, $N-1$ coded subfiles are transmitted. Therefore, a total of $KN(N-1)$ subfiles are getting transmitted in the delivery phase. 
\end{itemize}
The total delivery load $ R= KN(N-1)$ subfiles.  The normalized size of each of the coded subfiles is $\frac{1}{K(N-1)+1}$. Therefore the net rate is,
$$R = \frac{KN(N-1)}{K(N-1)+1} = N\left(1-\frac{1}{K(N-1)+1}\right).$$

The $(M,R)$ pair $(0,N)$ is achievable for a $(KN,N)$ coded caching scheme under any demand vector irrespective of its profile. Therefore, the rate of $R_{KN,N}(M) = N(1-M)$, $[0\leq M\leq \frac{1}{K(N-1)+1}]$ is achievable under uniform demand profile by memory sharing between $M=0$ and $M=\frac{1}{K(N-1)+1}$ appropriately.

This completes the proof of Lemma \ref{lemma1}. \end{proof}

\begin{thm}
\label{thm3}
         For a $(K,N)$ demand private coded caching scheme, the memory-rate pair $\left(\frac{1}{K(N-1)+1},N\left(1-\frac{1}{K(N-1)+1}\right)\right)$ is achievable. Furthermore, for $M\in [0, \frac{1}{K(N-1)+1}]$, we have
\begin{equation}
\label{eq4}
R_{K,N}^{p*}(M)=N(1-M).
\end{equation}
\end{thm}
\begin{proof}
Achievability of the rate $R_{K,N}^{p}(M) = N(1-M)$, $[0\leq M\leq \frac{1}{K(N-1)+1}]$  is straight forward. Lemma 1 shows the existence of a $(KN,N)$ non-private coded caching scheme that achieves the rate, $R_{KN,N}(M) = N(1-M)$ under uniform demand profile for $M\in[0, \frac{1}{K(N-1)+1}]$ . Therefore, by Theorem \ref{thm2} and Theorem \ref{thm2a}, the achievability of the same rate for a $(K,N)$ demand private coded caching scheme is guaranteed. Therefore, 
\begin{equation}
\label{eq5}
R_{K,N}^{p*}(M)\leq N(1-M) 
\end{equation}
for $0\leq M\leq \frac{1}{K(N-1)+1}$.

By the lower bound \eqref{lowerbound} given by Yan \textit{et al.} 
$$R^{p*}_{K,N}(M)\geq \max_{l\in[N]} \left(l+\frac{\min\{l+1,K\}(N-l)}{N-l+\min\{l+1,K\}}-lM\right)$$
\begin{equation}
\label{eq6}
R_{K,N}^{p*}(M)\geq N(1-M).
\end{equation}
Equation \eqref{eq6} is obtained by putting $l=N$ in the lower bound \eqref{lowerbound}. From \eqref{eq5} and \eqref{eq6}, we can conclude that for a $(K,N)$ demand-private coded caching scheme with cache memory  $M\in [0,\frac{1}{K(N-1)+1}]$, we have $R_{K,N}^{p*}(M)=N(1-M).$ \end{proof}

\begin{rem}
For $K=2, N=2$, the proposed scheme is same as the demand private scheme in \cite{c5} for the cache memory $M=\frac{1}{3}$.
\end{rem}

\subsection{Performance Comparison}
Here, we compare the performance of the proposed scheme with known schemes in \cite{c4}, \cite{c5}, \cite{c6} and \cite{c13} in terms of rate and subpacketization.
\begin{itemize}
\item Aravind \textit{et al.} shows the achievability of an $(M,R)$ pair $(\frac{2}{3},1)$ for a 2-file, 2-user demand-private coded caching with a subpacketization of 3 in \cite{c6}. That scheme can achieve the rate $\frac{3}{2}$ at $M=\frac{1}{3}$ by memory sharing. Whereas our proposed scheme achieves the rate $\frac{4}{3}$ at $M=\frac{1}{3}$ with the same subpacketization of 3. 
\item In \cite{c5}, authors gave the exact rate-memory trade-off for $(2,2)$ demand-private coded caching schemes and proved that $R_{2,2}^{p*}(\frac{1}{3})=\frac{4}{3}$ and can be achieved by a subpacketization of 3. Our work generalizes that particular memory-rate pair for a general $K$-user $N$-file demand-private coded caching system. That is, for the case of  2-user, 2-file demand-private coded caching, our proposed scheme achieves the rate $\frac{4}{3}$ at cache memory $M=\frac{1}{3}$ with same subpacketization of 3. 
\end{itemize}

In \cite{c4} and \cite{c13}, authors deal with general $K$-user, $N$-file demand-private coded caching systems. We compare the transmission rates of those schemes with our proposed scheme at cache memory $M=\frac{1}{K(N-1)+1}$.
\begin{itemize}
\item We showed the achievability of the rate $N(1-\frac{1}{K(N-1)+1})$ at a cache memory $\frac{1}{K(N-1)+1}$ with a subpacketization of $K(N-1)+1$. To fill a cache of size $\frac{1}{K(N-1)+1}$, a file has to be divided into at least $K(N-1)+1$ subfiles under uniform subpacketization.  
\item The virtual user scheme proposed in \cite{c4} achieves memory-rate pairs $(0,N)$ and $(\frac{1}{K},\frac{2K-N-1}{2K})$ with subpacketization 1 and $K$ respectively. So by memory sharing the rate ,
$$R = N\left(1-\frac{1}{K(N-1)+1}\right)+\frac{N-1}{2[K(N-1)+1]}$$
is achievable at memory $M=\frac{1}{K(N-1)+1}$ .We can readily observe that the rate is increased by an amount of $\frac{N-1}{2[K(N-1)+1]}$ compared to the rate achieved by our proposed scheme. To achieve this rate, a file has to be divided into $K$ subfiles of normalized size $\frac{1}{K(N-1)+1}$ and 1 subfile of normalized size $(1-\frac{K}{K(N-1)+1})$. The non-uniformity in the subpacketization is the consequence of memory sharing. The $(K,N)$ demand-private scheme in \cite{c4} is derived from the $(KN,N)$ non-private scheme in \cite{c2} where the placement is uncoded. By coded placement, we can lower the transmission rate especially in the lower memory regime. In addition to that, if we consider only the demands with uniform profile, a further reduction in delivery load is possible. By making use of these facts we are attaining the optimal rate for cache memory $M\leq \frac{1}{K(N-1)+1}$.
\item The \textit{LFR-DPCU} scheme proposed in \cite{c13} achieves the rate
\begin{equation*}
  R =
    \begin{cases}
       N\left(1-\frac{1}{K(N-1)+1}\right)+\frac{N-1}{K(N-1)+1} & \text{for $K\geq N$}\\
       N\left(1-\frac{1}{K(N-1)+1}\right)+\frac{K}{K(N-1)+1} & \text{for  $K<N$}\\
    \end{cases}       
\end{equation*}

at memory $M=\frac{1}{K(N-1)+1}$. For this memory, a file is divided into 2 subfiles of normalized sizes $(1-\frac{1}{K(N-1)+1})$ and $(\frac{1}{K(N-1)+1})$ to achieve the above rate. Even though the number of subfiles into which a file is divided is 2, the size of the smallest subfile still remains $\frac{1}{K(N-1)+1}$.  Our proposed scheme does better in terms of delivery load by paying a little in subpacketization compared to the demand-private scheme in \cite{c13} for cache memory $M\leq \frac{1}{K(N-1)+1}$.
\end{itemize}

\section{EXAMPLES}
\label{examples}


\begin{table*}[!htp]

  \centering
  \caption{The transmission $X^{d}$ for a particular choice of cache content and demand vector.}
  \label{tab1}
  \footnotesize\makegapedcells
   \adjustbox{max height=\dimexpr\textheight-5.5cm\relax,
   max width=\textwidth}{
  \begin{tabular}{|*{5}{c|}}
    \hline
    \multirowcell{5}[2ex]{\diagbox[height=5.0\line]{\raisebox{6ex}{\makecell{Content in\\Cache 1, Cache 2} }}{\raisebox{-3.5ex}{$X^{d}$}}} &\makecell{$B_1$\\$A_2$} & \makecell{$B_1$\\$A_2$} & \makecell{$A_1$\\$B_2$} & \makecell{$A_1$\\$B_2$}\\
          &$B_3$ &$A_3$ &$B_3$ &$A_3$\\
          &$A_1\oplus A_2\oplus A_3$ &$B_1\oplus B_2\oplus B_3$ &$A_1\oplus A_2\oplus A_3$ &$B_1\oplus B_2\oplus B_3$ \\ \hline
    $Z_1$,$Z_3$ &$[A,A]$ &$[A,B]$ &$[B,A]$ &$[B,B]$  \\ \hline
    $Z_1$,$Z_4$ &$[A,B]$ &$[A,A]$ &$[B,B]$ &$[B,A]$  \\ \hline
    $Z_2$,$Z_3$ &$[B,A]$ &$[B,B]$ &$[A,A]$ &$[A,B]$  \\ \hline
    $Z_2$,$Z_4$ &$[B,B]$ &$[B,A]$ &$[A,B]$ &$[A,A]$  \\
    \hline
    
  \end{tabular}
  }
\end{table*}

\subsection{Example 1 (N=K=2)}
We illustrate a $(4,2)$ non-private coded caching scheme achieving a rate $\frac{4}{3}$ under uniform demand profile at $M=\frac{1}{3}$. The construction of $(2,2)$ demand-private coded caching scheme from the $(4,2)$ non-private scheme is also shown. 
\begin{itemize}
\item Placement Phase: We call the two files in the system as $A$ and $B$. Each file is divided into 3 subfiles of equal size $A_1,A_2,A_3$ and $B_1,B_2,B_3$. We use  $G\in \mathcal{F}_2^{3\times 4}$  a generator matrix of a $(4,3)$-MDS code.
$$G = 
\begin{bmatrix} 
1 & 0 & 0 &1\\
0 & 1 & 0 &1\\
0 & 0 & 1  &1\\
\end{bmatrix}$$
Encode $[A_1, A_2, A_3]$ and $[B_1, B_2, B_3]$ using $G$. Then according to our previous notation, $\mathcal{C}_{1,1}=A_1$, $\mathcal{C}_{1,2}=A_2$, $\mathcal{C}_{1,3}=A_3$ and $\mathcal{C}_{1,4}=A_1\oplus A_2\oplus A_3$. Similarly, $\mathcal{C}_{2,1}=B_1$, $\mathcal{C}_{2,2}=B_2$, $\mathcal{C}_{2,3}=B_3$ and $\mathcal{C}_{2,4}=B_1\oplus B_2\oplus B_3$. Store $\mathcal{C}_{1,1}\oplus \mathcal{C}_{2,1}$ in cache 1. Populate all the caches in the similar fashion. And the cache contents are as follows,
		\begin{center}
			\begin{tabular}{|c|c|}
				\hline
				$Z_1$  & $A_1\oplus B_1$\\
				\hline	
				$Z_2$  & $A_2\oplus B_2$\\
				\hline
				$Z_3$  & $A_3\oplus B_3$\\
				\hline
				$Z_4$  & $A_1\oplus A_2\oplus A_3\oplus B_1\oplus B_2\oplus B_3$\\
				\hline
			\end{tabular}
		\end{center}

Let the demand vector be $\mathbf{d} = [A, A, B, B]$.
\item Delivery Phase: The server makes the transmissions: $B_1,B_2,A_3$ and $A_1\oplus A_2 \oplus A_3$. It can be verified that all the users get their requested file.
\end{itemize}
We have an example for $(4,2)$ non-private scheme achieving a memory rate pair $(\frac{1}{3},\frac{4}{3})$ (\cite{c12}).  Now, by the procedure given in \cite{c4}, we construct a $(2,2)$ demand-private coded caching scheme from the $(4,2)$ non-private scheme.

There are two users connected to the server. In user 1's cache, server stores either $Z_1$ or $Z_2$, with equal probability. Similarly, in second user's cache, server can store $Z_3$ or $Z_4$ during the placement phase. But the actual choice of the cache content is private between the server and the user. Suppose, server stores $Z_1$ in cache 1 and $Z_3$ in cache 2.  During the delivery phase, the actual demand vector $\mathbf{d}$ is mapped to the virtual demand vector $\mathbf{\tilde{d}}$. Assume that $\mathbf{d} = [ A ,B]$. Then $\mathbf{\tilde{d}} = [A ,B ,B ,A]$. The transmission is corresponding to $\mathbf{\tilde{d}}$. That is server sends $B_1$, $A_2$, $A_3$ and $B_1\oplus B_2 \oplus B_3$. From these transmissions, user 1 can decode the file $A$ and user 2 can decode $B$. Corresponding to the cache content and the demand vector, the virtual demand vector will vary. With the virtual demand vector, $X^{d}$ also varies. 

Table \ref{tab1} specifies the transmission $X^{d}$ for a given choice of cache content and demand vector. For example, if $Z_1$ and $Z_3$ are stored in cache 1 and cache 2 respectively, then for demand vector $\mathbf{d} = [ B,A]$, the transmissions are $A_1$, $B_2$, $B_3$ and $A_1\oplus A_2 \oplus A_3$. It can be seen that the user demands will remain private under this scheme. We can verify from Table \ref{tab1} that, different cache assignments lead to the same transmission corresponding  to different demand vectors. For instance, suppose server stores $Z_1$ and $Z_3$ in cache 1 and cache 2 respectively and assume that the demand vector is $[A,B]$. The transmissions are $B_1$, $A_2$, $A_3$ and $B_1\oplus B_2 \oplus B_3$. But the cache assignment $Z_1$, $Z_4$ also lead to the same transmission for the demand vector $[A,A]$. Since, the actual cache content is private, the user demands will also remain private.

\subsection{Example 2 (N=2, K=3)}
We illustrate a $(6,2)$ non-private coded caching scheme achieving a rate $\frac{3}{2}$ at $M=\frac{1}{4}$ under uniform demand profile. We call the two files in the system as $A$ and $B$. Here also, we assume that the file segments are coming from $\mathcal{F}_5\coloneqq \{0,1,2,3,4\}$.
\begin{itemize}
\item Placement Phase: Each file is divided into 4 subfiles of equal size $A_1,A_2,A_3,A_4$ and $B_1,B_2,B_3,B_4$. $G$ is a systematic generator matrix for a $(6,4)$ MDS code. Encode $[A_1, A_2, A_3, A_4]$ and $[B_1, B_2, B_3, B_4]$ using $G$. 
$$G = 
\begin{bmatrix} 
1 & 0 & 0 & 0 &1 &1\\
0 & 1 & 0 & 0 &1 &2\\
0 & 0 & 1 & 0 &1 &3\\
0 & 0 & 0 & 1 &1 &4\\
\end{bmatrix}\in \mathcal{F}_5^{4\times 6}$$
The coded subfiles are $\mathcal{C}_{1,1}=A_1$, $\mathcal{C}_{1,2}=A_2$, $\mathcal{C}_{1,3}=A_3$, $\mathcal{C}_{1,4} =A_4$, $ \mathcal{C}_{1,5} =A_1\oplus A_2\oplus A_3\oplus A_4$ and $ \mathcal{C}_{1,6} =A_1\oplus 2A_2\oplus 3A_3\oplus 4A_4$. Similarly, $\mathcal{C}_{2,1}=B_1$, $\mathcal{C}_{2,2}=B_2$, $\mathcal{C}_{2,3}=B_3$, $\mathcal{C}_{2,4} =B_4$, $ \mathcal{C}_{2,5} =B_1\oplus B_2\oplus B_3\oplus B_4$ and $ \mathcal{C}_{2,6} =B_1\oplus 2B_2\oplus 3B_3\oplus 4B_4$.
The cache contents are
		\begin{center}
			\begin{tabular}{|c|c|}
				\hline
				$Z_1$  & $\mathcal{C}_{1,1}\oplus \mathcal{C}_{2,1}$\\
				\hline	
				$Z_2$  & $\mathcal{C}_{1,2}\oplus \mathcal{C}_{2,2}$\\
				\hline
				$Z_3$  & $\mathcal{C}_{1,3}\oplus \mathcal{C}_{2,3}$\\
				\hline
				$Z_4$  & $\mathcal{C}_{1,4}\oplus \mathcal{C}_{2,4}$\\
				\hline
				$Z_5$  & $\mathcal{C}_{1,5}\oplus \mathcal{C}_{2,5}$\\
				\hline
				$Z_6$  & $\mathcal{C}_{1,6}\oplus \mathcal{C}_{2,6}$\\
				\hline
			\end{tabular}
		\end{center}

Let the demand vector is $\mathbf{d} = [A, A,A,B, B, B]$. The demand profile is $[3,3]$. That is, both the files are requested by 3 users.
\item Delivery Phase: Following transmissions are made: $\mathcal{C}_{2,1}$, $\mathcal{C}_{2,2}$, $\mathcal{C}_{2,3}$, $\mathcal{C}_{1,4}$, $\mathcal{C}_{1,5}$ and $\mathcal{C}_{1,6}$.
\end{itemize}
User 1, user 2 and user 3 are requesting for file $A$. They  directly get the coded subfiles $\mathcal{C}_{1,4}$, $\mathcal{C}_{1,5}$ and $\mathcal{C}_{1,6}$ of file $A$. In addition to that, user 1 can decode the coded subfile $\mathcal{C}_{1,1}$ as $\mathcal{C}_{1,1}=Z_1\ominus \mathcal{C}_{2,1}$. Since, every $4\times 4$ submatrix of $G$ is invertible, user 1 can decode the subfiles $A_1$, $A_2$, $A_3$ and $A_4$ from the 4 coded subfiles $\mathcal{C}_{1,1}$, $\mathcal{C}_{1,4}$, $\mathcal{C}_{1,5}$ and $\mathcal{C}_{1,6}$. Similarly, user 2 can get the coded subfile $\mathcal{C}_{1,2}$ and user 3 can get the coded subfile $\mathcal{C}_{1,3}$. So they can also get their demanded file $A$ from the 4 available coded subfiles. User 4, user 5 and user 6 are demanding for file $B$. All of them directly receive the coded subfiles $\mathcal{C}_{2,1}$, $\mathcal{C}_{2,2}$ and $\mathcal{C}_{2,3}$ of file $B$. From the local cache content, users can get one more coded subfile. That means, from the transmissions all the users get 4 encoded subfiles of the demanded file. From those subfiles, users can decode their required file. Demands of all the six users can be satisfied by the 6 transmissions given above. Since each transmission is $\frac{1}{4}$ of a file size, the normalized rate is $\frac{3}{2}$.

\section{future work}
In this work, we showed the achievability of a particular rate-memory pair for a general $K$-user, $N$-file coded caching problem satisfying the demand privacy requirements and also showed its optimality under a certain condition. The optimal rate-memory trade-off for a general $(K,N)$-demand private coded caching problem still remains open. Investigating different coded caching system models under demand privacy constraints is also very important.

\section{Acknowledgment}
This work was supported partly by the Science and Engineering Research Board (SERB) of Department of Science and Technology (DST), Government of India, through J.C. Bose National Fellowship to B. Sundar Rajan.

\end{document}